\documentclass[12 pt, onecolumn, draftclsnofoot]{IEEEtran}

\usepackage{setspace}
\usepackage[noadjust]{cite}
\usepackage{multirow}
\usepackage{dsfont}
\usepackage[latin1]{inputenc}
\usepackage{times}
\usepackage [spanish,english] {babel}
\usepackage [T1]{fontenc}
\usepackage [latin1]{inputenc}
\usepackage[dvips]{graphicx}
\usepackage{grffile}
\usepackage{amssymb,amsmath,amsthm, amsfonts}
\usepackage{xspace}
\usepackage{cases}
\usepackage{lineno}

\usepackage{algpseudocode}
\usepackage{algorithm}
\usepackage[nolist]{acronym}

\usepackage[usenames,dvipsnames]{color}

\usepackage{pifont}

\newtheorem{lemma}{Lemma}
\newtheorem{corollary}{Corollary}
\newtheorem{definition}{Definition}
\newtheorem{problem}{Problem}
\newtheorem{theorem}{Theorem}

\theoremstyle{remark}

\newtheorem{proofPart}{Part}

\newcommand{\DDC}[1][]{\ensuremath{D^{#1}(t)}}
\newcommand{\EEC}[1][]{\ensuremath{E^{#1}(t)}}
\newcommand{\acBatery}[2][t]{\ensuremath{B_A(#1;#2)}}
\newcommand{\acBateryIt}{\ensuremath{B_A^{(m)}(t)}}
\newcommand{\acData}[1][t]{\ensuremath{D_A(#1)}}
\newcommand{\acDataIt}{\ensuremath{D_A^{(m)}(t)}}
\newcommand{\minEExp}[2][t]{\ensuremath{E_{min}(#1;#2)}}
\newcommand{\minEExpIt}{\ensuremath{E_{min}^{(m)}(t)}}
\newcommand{\minDDep}[1][t]{\ensuremath{D_{QoS}(#1)}}
\newcommand{\minDDepIt}[1][t]{\ensuremath{D_{QoS}^{(m)}(#1)}}

\newcommand{\DArrivalT}[1][]{\ensuremath{d_i^{#1}}}
\newcommand{\EArrivalT}[1][]{\ensuremath{e_j^{#1}}}
\newcommand{\QArrivalT}[1][]{\ensuremath{q_k^{#1}}}
\newcommand{\DArrival}[1][]{\ensuremath{D_i^{#1}}}
\newcommand{\EArrival}[1][]{\ensuremath{E_j^{#1}}}
\newcommand{\QArrival}[1][]{\ensuremath{Q_k^{#1}}}

\newcommand{\noiseVar}[1][i]{\sigma^2}

\def \ie{, i.e.,\xspace}
\def \eg{, e.g.,\xspace}

\def \qos[#1]{\ac{QoS}#1}
\def \mimo[#1]{\ac{MIMO}#1}
\def \siso[#1]{\ac{SISO}#1}
\def \awgn[#1]{\ac{AWGN}#1}
\def \mmse[#1]{\ac{MMSE}#1}

\newcommand{\Nt}[1][n]{\ensuremath{N_t}}
\newcommand{\Nr}[1][n]{\ensuremath{N_r}}

\newcommand{\DomainX}[1][]{\ensuremath{\mathcal D_{\vec x}}}

\newcommand{\Ts}[1][j]{\ensuremath{T_s}\xspace}

\def\mb{\mathbf}

\def\vec#1{\mb{#1}}

\def\d{\textrm{d}}

\begin{document}








\begin{acronym}[OFDMA]

\acro{AWGN}{Additive White Gaussian Noise}
\acro{BCC}{Battery Capacity Constraint}
\acro{BER}{Bit Error Rate}
\acro{CSI}{Channel State Information}
\acro{CWF}{Classical Waterfilling}
\acro{DCC}{Data Causality Constraint}
\acro{DWF}{Directional Water-Filling}
\acro{EBS}{Empty Buffers Strategy}
\acro{ECC}{Energy Causality Constraint}
\acro{FSA}{Forward Search Algorithm}
\acro{ICT}{Information and Communications Technology}
\acro{ICTs}{Information and Communications Technologies}
\acro{IFFT}{Inverse Fast Fourier Transform}
\acro{i.i.d.}{independent and identically distributed}
\acro{ISO}{International Standards Organization}
\acro{ISS}{Incremental Slot Selection}
\acro{LT}{Luby Transform}
\acro{MAC}{Medium Access Control}
\acro{MAP}{Maximum a Posteriori}
\acro{MFSK}{Multiple Frequency-Shift Keying}
\acro{MGSS}{Maximum Gain Slot Selection}
\acro{MI}{Mutual Information}
\acro{MIMO}{Multiple-Input Multiple-Output}
\acro{ML}{Maximum Likelihood}
\acro{MMSE}{Minimum Mean-Square Error}
\acro{NDA}{Non Decreasing water level Algorithm}
\acro{OFDM}{Orthogonal Frequency Division Multiplexing}
\acro{OFDMA}{Orthogonal Frequency Division Multiple Access}
\acro{OSI}{Open System Interconnection}
\acro{PbP}{Pool-by-Pool}
\acro{QoS}{Quality of Service}
\acro{SISO}{Single-Input Single-Output}
\acro{SNR}{Signal to Noise Ratio}
\acro{SVD}{Singular Value Decomposition}
\acro{UPA}{Uniform Power Allocation}
\acro{WEHN}{Wireless Energy Harvesting Node}
\acro{WER}{Word Error Rate}
\acro{WF}{Waterfilling}
\acro{WFlow}{Water-Flowing}
\acro{w.r.t.}{with respect to}
\acro{WSNs}{Wireless Sensor Networks}
\acro{WSS}{Weighted Slot Selection}
\end{acronym}

\title{Energy-efficient transmission for wireless energy harvesting nodes}

\author{\IEEEauthorblockN{Maria Gregori and Miquel Payaró}\\
\IEEEauthorblockA{Centre Tecnol\`ogic de Telecomunicacions de Catalunya (CTTC)\\
E-mails: \{maria.gregori, miquel.payaro\}@cttc.cat\thanks{
This work was partially supported by: the Catalan Government under grants 2009SGR 1046 and
2011FI\_B 00956; the Spanish Ministry of Economy and Competitiveness under project TEC2011-29006-C03-01 (GRE3N-PHY); and the EC under project Network of Excellence in Wireless Communications (Newcom\#, Grant Agreement 318306).}}

\vspace{-0.7cm}
}
\maketitle
\thispagestyle{empty}
\doublespacing

\begin{abstract}

Energy harvesting is increasingly gaining importance
as a means to charge battery powered devices such as sensor
nodes. Efficient transmission strategies must be developed for \acp{WEHN} that take
into account both the availability of energy and data in the node. We consider a scenario where data and energy packets arrive to the node where the time instants and amounts of the packets are known (offline approach).
In this paper, the best data transmission strategy is found for a finite battery capacity \ac{WEHN} that has to fulfill some \ac{QoS} constraints, as well as the energy and data causality constraints. As a result of our analysis, we can state that losing energy due to overflows of the battery is inefficient unless there is no more data to transmit and that the problem may not have a feasible solution. Finally, an
algorithm that computes the data transmission curve minimizing
the total transmission time that satisfies the aforementioned constraints has been developed.

\end{abstract}

\newpage
\pagenumbering{arabic}

\setcounter{page}{1} 
\section{Introduction}




Battery autonomy is limiting the operational lifetime of battery powered wireless nodes due to the  increasing computational capabilities, which follow from the reduction of transistor's size, and the consequent increment of the node throughput requirements. 
As battery capacity is not growing as fast as the energetic demand of battery powered wireless nodes, e.g., handheld devices or sensor nodes, it is essential to make an efficient use of the available energy.

The energy consumption to transmit a certain amount of
data can be reduced by increasing the transmission time or, equivalently, reducing the transmitted power or rate at which the information is sent \cite{zafer_calculus_2009}.
However, there exists a trade-off between energy and transmission delay as in many real scenarios, the data flow must guarantee certain levels of performance.
For instance, the network may require that data is transmitted with a certain minimum rate, quality level or with a maximum delay. These requirements, usually referred to as \ac{QoS} constraints, become crucial in real-time applications that are sensitive to the network delay.

The trade-off between energy and transmission delay was studied in different
works \cite{prabhakar_energy-efficient_2001, zafer_delay-constrained_2007, zafer_calculus_2009,
uysal-biyikoglu_adaptive_2002, uysal_energy-efficient_2002}, where the objective was to minimize energy
consumption while satisfying \ac{QoS} constraints. Specially remarkable is the work by Zafer et al. in \cite{zafer_calculus_2009}, as they proposed a novel solution that used the concept of cumulative curves, which allows a graphical and simple
visualization of the solution.

Within this context, energy harvesting\ie the process of collecting energy from the environment by different means\eg a solar cell or a piezoelectric generator \cite{Paradiso_EH_wirelessElec_2005, Chalasani_Survey_2008}, is becoming an appealing solution to enlarge the autonomy of battery powered devices.
Energy harvesting has opened a new research paradigm on the design of energy-efficient transmission strategies as the transmission strategies for non-harvesting nodes are no longer optimal.
Many works\eg  \cite{ yang_optimalPacket_2010,tutuncuoglu_optimum_2010, ho_optimal_2011,Ozel_Ulukus_Yener_TX_Fading_Optimal_2011,Yang_Broadcasting_2010,Akif_Broadcasting_2011,ozel_BC_finiteBattery_2012,Ozel_Information_2010,Sharma_Management_EH_2010}, model the  energy harvesting process as a set of energy packets arriving to the node at different time instants and with different amounts of energy.
There exist two well established approaches for the design of optimal transmission strategies, namely,  \textit{online} and \textit{offline}.
The \textit{online} approach assumes that the node only has some statistical knowledge of the dynamics of the energy harvesting process.
The \textit{offline} approach assumes that the node has full knowledge of the amount and arrival time of each energy packet, which is an idealistic situation that provides analytical and intuitive solutions that can be used to gain insight for the later design of the \textit{online} transmission strategy. 

By using this packetized model for the energy harvesting process,  \cite{yang_optimalPacket_2010} considered dynamic data packet arrivals
and found the transmission strategy that minimizes the delivery time of all data packets, however, the authors of \cite{yang_optimalPacket_2010} assumed an infinite battery capacity.
In \cite{tutuncuoglu_optimum_2010}, a node with finite battery capacity was studied, however, considering that all the data packets are available from the beginning of the transmission. 
This assumption  significantly simplifies the setup as losing energy due to battery overflows is clearly suboptimal because there is always data to be transmitted.
References \cite{ho_optimal_2011} and  \cite{Ozel_Ulukus_Yener_TX_Fading_Optimal_2011} used convex optimization and dynamic programming to find the best rate scheduling or power allocation strategy under the presence of fading. 
The minimization of the transmission completion time for a \ac{WEHN} operating in a broadcast link was considered in \cite{Yang_Broadcasting_2010,Akif_Broadcasting_2011,ozel_BC_finiteBattery_2012}.
 References \cite{Yang_Broadcasting_2010} and \cite{Akif_Broadcasting_2011} assumed infinite battery capacity, and, the authors of \cite{ozel_BC_finiteBattery_2012} found the rate scheduling policy of the finite battery capacity case. 
However, these works assumed that the data to be transmitted to each of the users is available from the beginning at the transmitter, similarly as in the point-to-point case in \cite{tutuncuoglu_optimum_2010}.
In \cite{Ozel_Information_2010}, the coding problem was studied from an information theory perspective for an energy harvesting node.
Finally, the stability and delay of the data queue were considered in \cite{Sharma_Management_EH_2010} to derive the optimal transmission policies.

To the best of our knowledge, there is no work that considers a \ac{WEHN} with finite battery capacity and dynamical data and energy arrivals. 
In a prior work \cite{gregori_globecom}, we studied this problem and showed that as far as battery does not overflow, constant rate transmission is the strategy that requires less energy to transmit a certain amount of data.
However, if battery overflows, transmitting at constant rate is no longer optimal, but the optimal strategy increases the rate before the overflow until either there is no overflow or until all the data is transmitted.
In this paper, we extend the work in \cite{gregori_globecom} to take into account new \ac{QoS} requirements of the different data packets, which was not considered in \cite{gregori_globecom}.
\ac{QoS} constraints substantially complicate the problem as we have constraints both in the data and energy domains. 
In this paper, we propose a framework to map the constraints of the energy domain to the data domain that allows us to adapt the calculus approach proposed in \cite{zafer_calculus_2009}, which did not take into account energy harvesting at the transmitter, to the energy harvesting scenario. Therefore, the main contributions of this paper are:
\begin{itemize}
	\item Studying the impact of the \ac{QoS} constraint in the transmission strategy that minimizes the transmission completion time, which to the best of our knowledge has not been previously considered in the literature. 
	\item Showing that, due to the \ac{QoS} constraints, there may be situations in which no feasible transmission strategy exists. Such situations are analytically characterized in Lemma \ref{L:Solution}.
	\item Showing that if the optimal transmission strategy exists, the optimal cumulative data departure curve is
a peace-wise linear function where battery overflows are only produced when the data buffer is empty.
	\item Developing an algorithm that either computes the optimal transmission strategy or concludes that there is no feasible solution and analytically showing its optimality.
\end{itemize}

In Section \ref{sec:Problem}, the problem is mathematically formulated  by using the concept of cumulative curves, which is introduced in \cite{zafer_calculus_2009}, that allows an appealing visualization of the solution.
The solution is characterized in Section \ref{sec:QoS}.
Section \ref{sec:algorithm_qos}, presents the developed iterative algorithm that is able to compute the optimal solution by using the characterization performed in Section \ref{sec:QoS}.
Finally, the conclusions of the work are presented in Section \ref{sec:conclusions}.

\vspace{-0.3cm}
\section{Problem formulation}
\label{sec:Problem}
We consider a node with a finite battery capacity, $C_{max}$, that has to
transmit $N$ data packets by using at the most the $J$ energy packets that it harvests over
time while satisfying some \ac{QoS} requirements (see Figure \ref{fig:sysModel_QOS}). We want to find the power allocation/rate scheduling strategy\footnote{Observe that fixing the transmission power or the rate is equivalent as it will be shown next.} that minimizes the transmission completion time, $T$.

\begin{figure}
\centering
\includegraphics[width=0.7\columnwidth]{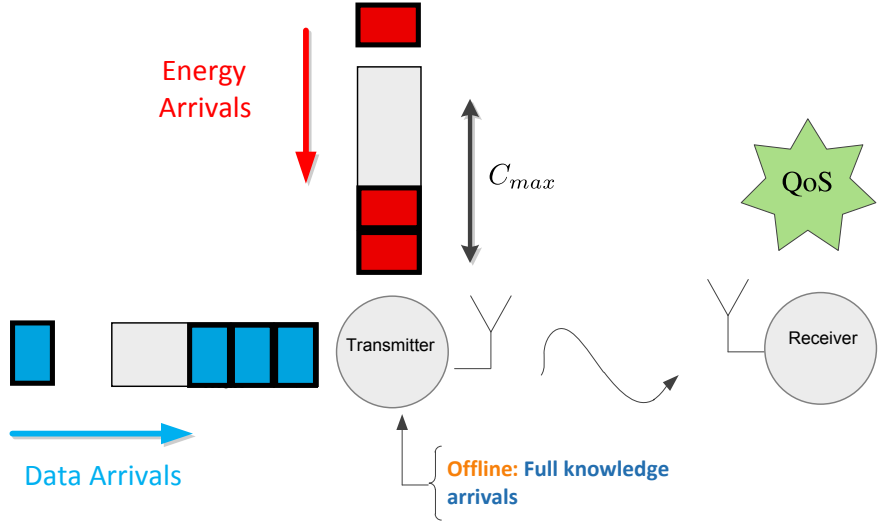}
\caption{System Model.}
\label{fig:sysModel_QOS}
\end{figure}

We assume that the time instants at which the data and energy packets arrive to the
node and their size (bits or Joules) are known from beforehand (\textit{offline} approach). Hence, it is known that at the
time instant $d_i \geq 0$ seconds the $i$-th data packet arrives containing $D_i$ bits,
with $i=0\ ...\ N-1$. Similarly, the $j$-th energy packet arrives at
the instant $e_j \geq 0$ seconds and a total of $E_j$ Joules are harvested, with $j=0\
...\ J-1$ (see Figure \ref{fig:timeDomain}). Without loss of generality, the first energy arrival is produced at $e_0=0$ and contains the initial battery of the node $E_0$.

\begin{figure}
\centering
\includegraphics[width=0.7\columnwidth]{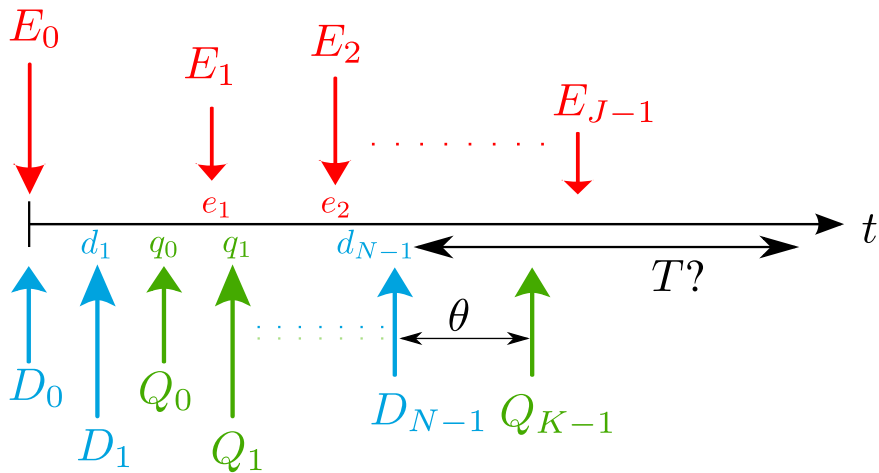}
\caption{Summary of the different events considered in the time domain, namely, energy arrivals, data arrivals, and quality of service requirements. The figure represents the deadline \ac{QoS} constraint, where $\theta$ is the maximum allowed delay for all the packets.}
\label{fig:timeDomain}
\end{figure}

To describe our model we present the following definitions that are summarized in Table \ref{tab:Notation}:

\begin{definition}[Data Departure Curve]
A data departure curve $D(t) \geq 0$, $t \geq 0$, is the total number of bits that have been cumulatively transmitted by the node in the time interval $[0,t]$.
\label{Def:D}
\end{definition}

\begin{definition}[Energy Expenditure Curve]
An energy expenditure curve $E(t) \geq 0$, $t \geq 0 $, is the energy in Joules that has been cumulatively consumed by the node in the time interval $[0,t]$.
\label{Def:E}
\end{definition}

Let us consider a static or slow-fading channel where the power-rate function
$g(\cdot)$, i.e., the function that, at any given time instant $t$, relates the
transmitted power, $P(t)$, with the rate, $r(t)$, according to $P(t) =
g(r(t))$. As in \cite{zafer_calculus_2009} and \cite{yang_optimalPacket_2010}, we make the common assumption that the function $g(\cdot)$ is
time-invariant, convex, monotonically increasing, and only depends on $r(t)$.
Note that the instantaneous rate, $r(t)$, can be expressed as the derivative
with respect to $t$ of the data departure curve, i.e., $r(t) =\frac{\d D(t)}{\d t}$.
Similarly, the transmitted power is $P(t) =\frac{\d E(t)}{\d t}$.
Then, the energy
expenditure curve can be obtained from the data departure curve as follows\footnote{
Without loss of generality, we can assume that $r(t)$ is right-continuous.}: 
\begin{equation}
E(D(t)) = \int_0^t g\left(\frac{\d D(\tau)}{\d \tau}\right) \: \d \tau.
\label{EQ: EfromD}
\end{equation}

Observe that the magnitudes $D(t)$, $E(t)$, $r(t)$, and $P(t)$ are
unambiguously related by \eqref{EQ: EfromD} and  $g(\cdot)$. Therefore, given
the initial states 
$E(0) = 0$ and $D(0) = 0$, 
the design of the system to be optimized can be described by
any of these magnitudes.

\begin{definition}[Battery] 	
The battery of the node $B(t)$ is the amount of energy that the node has
available at a given time instant $t$. We consider a battery with finite
capacity $C_{max}$. Thus, $B(t) $ must satisfy that $0 \leq B(t) \leq C_{max}$,
$\forall t \geq 0$.
\end{definition}

Due to the limited battery capacity, at the $j$-th energy arrival, some part of the harvested energy $E_j$ may be lost. This  lost energy is denoted as the $j$-th battery overflow\ie $O_j =\{ E_j - C_{max} + B(e_j^-) \}^+$, with $\{x\}^+  = \max\{0,x\}$. Observe that battery overflows depend on the chosen energy expenditure curve $E(t)$ and guarantee that the battery level will never be above the battery capacity, $B(e_j^+) = B(e_j^-) + E_j - O_j \leq C_{max}$.

As battery overflows depend on the chosen $E(t)$, their value $O_j$ cannot be computed until the energy expenditure curve $E(t)$ is fixed  $\forall t \leq e_j$. Loosely speaking, one can expect that the optimal solution uses efficiently the harvested energy to transmit the available data and, at the same time, tends to minimize the
total overflow of the battery and, thus, maximizes the accumulated energy stored in the battery. In
other words, if overflows are minimized, the node will be able to use more
energy, as, otherwise, the energy of the overflows is lost. 

In the following lines, we define the accumulated battery, a concept introduced in
this work that allows us to characterize the optimal solution when having finite
battery capacity constraint. Let $t_x$ denote the last time instant up to which $D(t)$ (or, equivalently, $E(t)$) is known.
\begin{definition}[Accumulated Battery]
The accumulated battery $B_A(t;t_x)$ is a real measure of the accumulated energy stored in the battery for $t \in [0, t_x]$ and it is the maximum possible accumulated energy in the battery for $t \in (t_x, \infty)$ (assuming that no overflows are produced for $ t \in (t_x,\infty)$)\ie
$B_A(t;t_x) = \sum_{j=0}^{J-1} \left(-O_jH(t_x-e_j) + E_j \right)H(t-e_j)$, where $H(\cdot)$ is the Heaviside function.
\end{definition}
Observe that, for $t \in [0,t_x]$, $B_A(t;t_x)$, represents the real measure of the energy accumulated in the battery during $(0,t)$, because, for these time instants, overflows are known and taken into account. Alternatively, for $t \in (t_x,\infty)$, battery overflows are unknown after $t_x$ and $B_A(t;t_x)$ models the best case scenario where the node is able to store in the battery all the  harvested energy in the interval $(t_x,\infty)$. Moreover, observe that $B_A(t;t_x)$ = $B_A(t; t_y)$ for any $t_x, t_y \geq t$.


At every time instant, the energy stored in the battery is $B(t)=  B_A(t;t) - E(t)$. Note that $B_A(t;t)$ takes into account the actual net incoming energy in the battery, whereas, $E(t)$ is the net outgoing energy. Thus, their difference results in the energy stored in the battery.

\begin{definition}[Minimum Energy Expenditure]
The minimum energy expenditure, $E_{min}(t; t_x)$, is the smallest amount of energy
that the node must have cumulatively spent at time $t > t_x$ such that no overflow of the battery
is produced in the interval $(t_x,t]$\ie $E_{min}(t;t_x) = \left\{ B_A(t;t_x) - C_{max}  \right\}^+ $.
\end{definition}

\begin{definition}[Accumulated Data]
The accumulated data $D_A(t)$ is the sum of data that has arrived at the node
during the time interval $[0,t)$, i.e., $D_A(t)=\sum_{i=0}^{N-1}D_i H(t-d_i)$.
\end{definition}

Different \ac{QoS} constraints can be considered by mapping the constraint into an appropriate minimum data departure curve, which
was introduced in \cite{zafer_calculus_2009}, and is defined as follows:
\begin{definition}[Minimum Data Departure]
\label{Def:min_data_dep}
The minimum data departure, $D_{QoS}(t)$, is the smallest amount of data
that the node must have cumulatively transmitted at time $t$ such that the  \ac{QoS} constraint is satisfied.
\end{definition}

The rest of the paper considers a general $D_{QoS}(t)$ that is a non-negative  staircase function where changes are produced at time instants $q_k$ with increments of $Q_k$ bits for $k=0\dots K-1$\ie $D_{QoS}(t) = \sum_{k=0}^{K-1}Q_k H(t-q_k)$. From now on, the instants $q_k$ are called \textit{quality requirement events}.
Thus, three kind of events are considered, namely, data arrival, energy arrival, and quality requirement events, as summarized in Figure \ref{fig:timeDomain}. Let us briefly describe two examples, introduced in \cite{zafer_calculus_2009}, of how to map a certain \ac{QoS} constraint to the minimum data departure curve:

\begin{itemize}
\item Deadline Constraint: This constraint considers that the maximum permissible delay for the transmission of a certain data packet, $D_k$, is $\theta_k$ seconds. Then, $D_{QoS}(t)$ is a piecewise constant function that changes at instants $q_k= d_k + \theta_k$ with an increment of $D_k$. As a specific case, we can consider that the allowed delay for all the packets is the same, i.e., $\theta_k = \theta$, $\forall k$. Then, the minimum data departure is given by $D_{QoS}(t)= D_A(t-\theta)$.

\item Buffer constraint: A limited data queue of size $\beta$ is considered. Hence, in order not to loose any incoming data, the minimum data departure must be $D_{QoS}(t) = \{ D_A(t)- \beta \}^+$.

\end{itemize}
Observe that depending on the chosen \ac{QoS}, it is likely that the instants $q_k$ are equal to $d_i$ for some values of $k$ and $i$, e.g., for the buffer-type constraint. Hence, we consider that two different types of events can be produced simultaneously at the same time instant.\footnote{
Because of this possible simultaneity between events, the problem considered in this paper is much more difficult than the one considered in \cite{gregori_globecom}.}

Our goal is to find the data departure curve, $D(t)$, that minimizes the
transmission completion time $T$ of the $N$ data packets\ie $D(T)=\sum_{i=0}^{N-1}D_i$, while satisfying the following conditions:
 $(i.)$ \ac{ECC}: energy must be
harvested before it is used by the node or, which is the same, the battery level
in the node must be greater or equal to zero. 
$(ii.)$ \ac{DCC}: it is not
possible to transmit more bits than the ones that have arrived to the node.
$(iii.)$ \ac{QoS} constraint: at time $t$, a minimum amount of data $D_{QoS}(t)$ has to be transmitted in order to preserve the link quality of service.
Moreover, given two data departure curves with the same completion time, the one
that requires less energy is always preferred. From all that has been said above,
the problem can be expressed as follows:
\begin{IEEEeqnarray}{rCl}
\label{Eq:P_EA_DA_Cmax_QOS}
\min_{\DDC} \quad  & &  T  \\
s.t. \:  \quad& & E(t) \leq B_A(t;t), 				\nonumber			\\
	& &	D_{QoS}(t)  \leq D(t) \leq  D_A(t),	 \nonumber	\\
	& &	D(T) = \sum_{i=0}^{N-1}D_i			\nonumber		.
\end{IEEEeqnarray}

We want to remark the two main difficulties of the problem presented in \eqref{Eq:P_EA_DA_Cmax_QOS}.
First, the integral relation among the data and energy domains through \eqref{EQ: EfromD}.
Second, the fact that neither $T$ nor $B_A(t; t)$ are known from beforehand, due to their dependence on $\DDC$.
Consequently, both $T$ and $B_A(t; t)$ will be found along with the solution to the problem.

This problem is graphically
represented in Figure \ref{fig:P_Visual_QOS}, where the figures at the top and bottom
stand for the energy and data domains, respectively. The \ac{ECC}
is 
represented by
the solid line at the top figure, whereas \ac{DCC} and QoS constraints are
depicted
by the dot-dashed and dashed lines in the figure at the bottom, respectively. The dotted line in the energy domain represents the minimum energy expenditure curve, however, as it can be seen from \eqref{Eq:P_EA_DA_Cmax_QOS}, it is not a constraint of the problem. Hence, $D(t)$ and its associated $E(t)$ must lie within the blank region of the data and energy domains, respectively, in order to be a valid solution. Three different data departure curves (A, B, and C) and their associated
energy expenditure curves are shown. The curve A is not valid since it breaks the \ac{ECC}. The curve B is valid in spite of having an overflow of the battery at $e_3$, which is, in general, a suboptimal strategy as we will show later. Finally, the curve C is not valid because it does not satisfy the \ac{QoS} constraint.

\begin{figure}
\centering
\includegraphics[width=0.55\columnwidth]{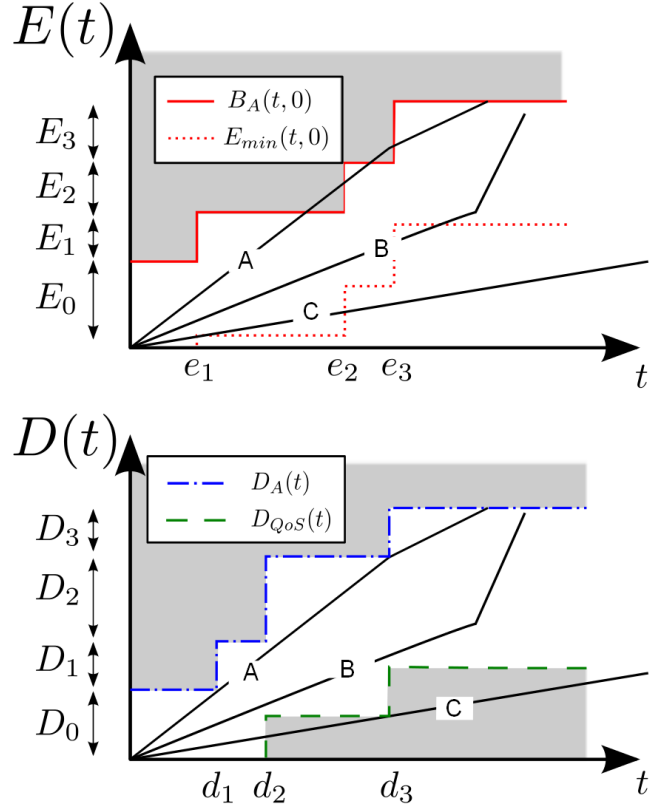}
\vspace{-1em}
\caption{Illustrative representation of the problem presented in \eqref{Eq:P_EA_DA_Cmax_QOS}.}
\label{fig:P_Visual_QOS}
\end{figure}

We want to remark that Figure \ref{fig:P_Visual_QOS} is just an illustrative representation of the problem. As mentioned before, both the accumulated battery and the minimum energy expenditure depend on the selected transmission strategy. Note that the figure shows $B_A(t; 0)$ as \ac{ECC} instead of the real \ac{ECC} $B_A(t; t)$. In some manner, with $B_A(t; 0)$, we are showing the battery accumulated that is obtained when $D(t)$ does not produce any battery overflow. In case that overflows cannot be avoided, it would be necessary to subtract the amount of the overflow to the shown accumulated battery and minimum energy expenditure from that time instant onward. Thus, we can only show the real graphical representation of the problem once we have fixed the solution.


Note that \eqref{Eq:P_EA_DA_Cmax_QOS} is not a convex optimization problem and that its conversion into a convex problem is not straightforward.
Thus, we cannot directly solve \eqref{Eq:P_EA_DA_Cmax_QOS}.
Alternatively, in next section, we model the properties that the optimal solution must satisfy, which will allow us to construct the optimal data departure curve.

\section{Properties of the optimal solution}
\label{sec:QoS}
As pointed out previously, in this section we will characterize the optimal data departure curve, $D^\star(t)$, and its associated energy expenditure curve, $E^\star(t)$, for the problem \eqref{Eq:P_EA_DA_Cmax_QOS}:


\begin{problem}[Transmission Without Events]
Consider that the optimal departure curve is known up to $t_1$ and that we want to characterize its optimal behaviour in the time interval $(t_1
, t_2)$ where there are no changes in $D_A(t)$, $B_A(t; t_1)$, and $D_{QoS}(t)$ . We also consider that
the data departure curve at the boundary of the intervals is $D(t_1)$ and
$D(t_2)$, respectively, and that these two points satisfy the data, energy, and QoS constraints.
\label{P:NoEvents_QOS} 
\end{problem}

\begin{lemma}
In Problem \ref{P:NoEvents_QOS}, $D^\star(t)$ is a straight line where the slope,
or, equivalently, the transmission rate, is constant and equal to $r(t) =
\frac {D(t_2) - D(t_1)}{t_2 - t_1} $, $\forall t \in (t_1, t_2) .$
\label{L:NoEvents_QOS}
\end{lemma}

\begin{proof}
The proof follows from the integral version of Jensen's inequality in a similar way to the BT-problem in \cite{zafer_calculus_2009}.
\end{proof}


\begin{corollary}
Lemma \ref{L:NoEvents_QOS} implies that $D^\star(t)$ is a piece-wise linear
function such that its slope, which is equivalent to the transmission rate, can
only change either at $e_j$, $d_i$ or $q_k$. 
\label{C:Piecewise_QOS}
\end{corollary}

From the previous lemma and corollary, it follows that constant rate transmission saves energy due to the convexity of the power-rate function $g(\cdot)$.
However, constant rate transmission is not optimal when a battery overflow is produced because the energy saved due to constant rate transmission is lower than the energy lost in the overflow. Consequently, the optimal solution increases the rate before the overflow until either there is no overflow or the data buffer is empty,
 as shown in the following lemma and its subsequent proof:
\begin{lemma}
\label{Result:NoOverflow}
Under the optimal policy, battery overflows may only be produced when there is no data to be transmitted.
\end{lemma}
\begin{IEEEproof}
A similar proof was given in \cite{gregori_globecom} for the \ac{QoS} relaxed problem. 
Note that Lemma \ref{Result:NoOverflow} states that if an overflow is produced, the data departure curve must reach its upper bound\ie $\DDC = \acData$. Thus, the insertion of the \ac{QoS} constraint, which is a lower bound 
on $D(t)$, does not affect the validity of the proof given in \cite{gregori_globecom}.
Therefore, Lemma \ref{Result:NoOverflow} can be proved by following the approach in \cite[Appendix A]{gregori_globecom}.
\end{IEEEproof}

 The following lemma states that by the end of the transmission the battery must be empty, otherwise, transmission could have been finished earlier by transmitting at a higher rate.

\begin{lemma}
\label{L:useAllEnergy}
The optimal solution must satisfy that, at the instant $T$ at
which all the data has been transmitted, the energy expenditure is equal to the accumulated
battery, i.e., $E^\star(T) = B_A(T; T)$.
\end{lemma}
\begin{IEEEproof}
Similar to the proof of Lemma 5 in \cite{tutuncuoglu_optimum_2010}.
\end{IEEEproof}

In the remainder of the paper, the term \textit{epoch} denotes each of the time intervals of  $D^\star(t)$ at which transmission is done at constant rate.
We define $M$ as the total number of epochs of the optimal solution, i.e., the number of linear pieces of $D^\star(t)$.
Note that $M$ is unknown a priori.
Consequently, solving the problem in \eqref{Eq:P_EA_DA_Cmax_QOS} is equivalent to determining the rate $r_m$ and length $\ell_m$ of each epoch\ie $\{r_m, \ell_m\}_{m=0}^{m=M-1}$.
To do so, we have developed an iterative algorithm that, at the $m$-th iteration, determines the rate and duration of the $m$-th epoch\ie $\{r_m,\ell_m\}$.
We denote $\tau_m$ as the instant at which the $m$-th epoch begins\ie $\tau_m = \sum_{p=0}^{m-1} \ell_p$.
The algorithm ends when all data has been efficiently transmitted. With this, $M$, $T$, $B_A(t; T)$, and $D^\star(t)$ are found.

To simplify the complexity of our algorithm, at the beginning of the $m$-th iteration\footnote{Note that $D^\star(t)$ is known in $[0,\tau_m]$.}, the origin of coordinates is moved to the point $(\tau_m, D^\star(\tau_m))$. To be coherent with the vertical and horizontal displacement of the origin of coordinates, the data and energy constraints in \eqref{Eq:P_EA_DA_Cmax_QOS} must be vertically rescaled by $D^\star(\tau_m)$ or $E^\star(\tau_m)$, respectively, and temporally displaced by $\tau_m$. In the remainder of the paper, a super-index $(m)$ above a variable ($B_A^{(m)}(t)$, $D_A^{(m)}(t)$, $D_{QoS}^{(m)}(t)$, and $E_{min}^{(m)}(t)$) denotes that it is the rescaled version at the $m$-th iteration\footnote{The relations among the iteration specific and general versions of the variables are given in Table \ref{tab:Notation}.}, e.g., $B_A^{(m)}(t) = B_A(t+\tau_m; \tau_m) - E^\star(\tau_m)$\footnote{We have dropped the second argument in $B_A^{(m)}(t)$ and $E_{min}^{(m)}(t)$ since within the $m$-th iteration, the second argument, which denotes the last instant at which the solution is known, is always $\tau_m$.}.

From the structure of the problem in \eqref{Eq:P_EA_DA_Cmax_QOS}, it is easy to expect that it may not have a feasible solution whenever the node has to fulfill very tight \ac{QoS} requirements, while, at the same time, it does not harvest enough energy to transmit all the required data. The following lemma is checked  at every iteration of our proposed algorithm to determine whether the problem in \eqref{Eq:P_EA_DA_Cmax_QOS} has a feasible solution or not.
\begin{lemma}
\label{L:Solution}
The problem \eqref{Eq:P_EA_DA_Cmax_QOS} does not have a feasible solution whenever
\begin{equation}
\label{Eq:SolutionExists}
D_{QoS}^{(m)}(q_k^{(m)}) > q_k^{(m)} g^{-1}\left(B_A^{(m)}(q_k^{(m)})/q_k^{(m)}\right),
\end{equation}
for some quality requirement event $q_k^{(m)} \in (0, T^{(m)})$\footnote{\label{FN:ReescaledTV}Where $X^{(m)} = X- \tau_m$ is rescaled version of some temporal variable $X$ at the $m$-th iteration (see Table \ref{tab:Notation}).}.
\end{lemma}

\begin{IEEEproof}
Whenever we encounter that the problem does not have a feasible solution is because, at some quality requirement event $q_k^{(m)}$, it is not possible to fulfill all the constraints. Let $\bar D(t)$ be the data departure curve that transmits the maximum amount of data in the interval $(0,q_k^{(m)})$\ie  $\bar D(t) = g^{-1}\left(B_A^{(m)}(q_k^{(m)})/q_k^{(m)}\right) q_k^{(m)}$. Note that this curve has constant rate, empties the battery  at $q_k^{(m)}$, and that the constrains are not necessarily satisfied. If, at some $q_k^{(m)} \in (0, T^{(m)})$, the \ac{QoS} constraint requires more than $g^{-1}\left(B_A^{(m)}(q_k^{(m)})/q_k^{(m)}\right) q_k^{(m)}$ bits to be transmitted, then the problem does not have a feasible solution.
\end{IEEEproof}


If there is no $q_k^{(m)} \in (0, T^{(m)})$ that satisfies Lemma \ref{L:Solution}, then the problem still may have a feasible solution and at least another epoch\ie $\{r_m,\ell_m\}$, can be determined. In this context, in next subsection, we model how the rate changes must be produced in order to be optimal.


\subsection{Constraints mapping into the data domain for a given epoch}
\label{Sec:DataMapping_QOS}

Within an algorithm iteration, the \ac{ECC} can be mapped to the data domain, hence, allowing us to merge both constraints to the most restrictive constraint.
Let us consider that the algorithm is at the beginning of the $m$-th iteration\ie the optimal solution is known up to $\tau_m$, where the rate $r_m$ and length $\ell_m$ of the $m$-th epoch must be determined.
Given that transmission must be done at constant rate/power, the maximum amount of data that can be transmitted at a certain time instant $t_y$ due to the energy causality constraint is $g^{-1}(p_y)t_y$, where $p_y = B_A^{(m)}(t_y)/t_y$.\footnote{
Note that for the constraints mapping it is not necessary that $p_y$ satisfies the constraints. The constraints fulfillment  is enforced by the algorithm that computes the optimal solution, which is explained in Section \ref{sec:algorithm_qos}.}
With this, at the instant $t_y$,  the upper-bound on the energy expenditure curve has been mapped to an upper-bound on the data departure curve, as shown in Figure \ref{fig:mappings} for the instants $t_1$ and $t_2$.
By applying this procedure at all time instants $t_y \in (0, T^{(m)})$, we can map the whole upper-bound in the energy domain to an upper-bound in the data domain, which we denote by $\bar D_{B_A}^{(m)}(t) = g^{-1}(B_A^{(m)}(t)/t)t$ and call \textit{actual mapping}.
However, doing this computation for each time instant has a high computational cost.
In summary, if a data departure curve that transmits at constant rate\ie $D^{(m)}(t) = r_m t$, satisfies that $D^{(m)}(t)  \leq \bar D_{B_A}^{(m)}(t), \forall t$, then it also satisfies the \acp{ECC}.

The cost associated with the computation of $\bar D_{B_A}^{(m)}(t)$ can by reduced by noting that it is suboptimal that $D^{(m)}(t)$ reaches the \textit{actual mapping} at any time instant $t_1$ that is not an event\ie  $t_1 \neq d_i^{(m)}$, $t_1 \neq e_j^{(m)}$, and   $t_1 \neq q_k^{(m)}$, $\forall i, j, k$ . This is clearly seen in Figure \ref{fig:mappings}. Observe that
if $D^{(m)}(t_1) = \bar D_{B_A}^{(m)}(t_1)$, then the battery is empty at $t_1$\ie  $E^{(m)}(t_1) = B_A^{(m)}(t_1)$, which follows from the definition of the actual mapping. Consequently, in order to satisfy \acp{ECC}, the rate at $t_1^+$ must be zero as no energy arrival is produced at $t_1$. From Corollary \ref{C:Piecewise_QOS}, we know that this rate change is suboptimal and, therefore, the data departure curve can only reach the \textit{actual mapping} in some event. Thus, to reduce the computational complexity of the \textit{actual mapping}, we can compute the value of the mapping only at the aforementioned events and assign a constant value in the interval between events. In the rest of the paper, we refer to this mapping as \textit{effective mapping}, i.e., $D_{B_A}^{(m)}(t)$. Note that the \textit{effective mapping} is an upper bound of the \textit{actual mapping}. However, we want to remark that by using the \textit{effective mapping}, we are not relaxing the constraints of the problem since both mappings are equal at the time instants where the optimal solution for the data departure curve coincides with the \textit{actual mapping}.

\begin{figure}
\centering
\includegraphics[width=0.50\columnwidth]{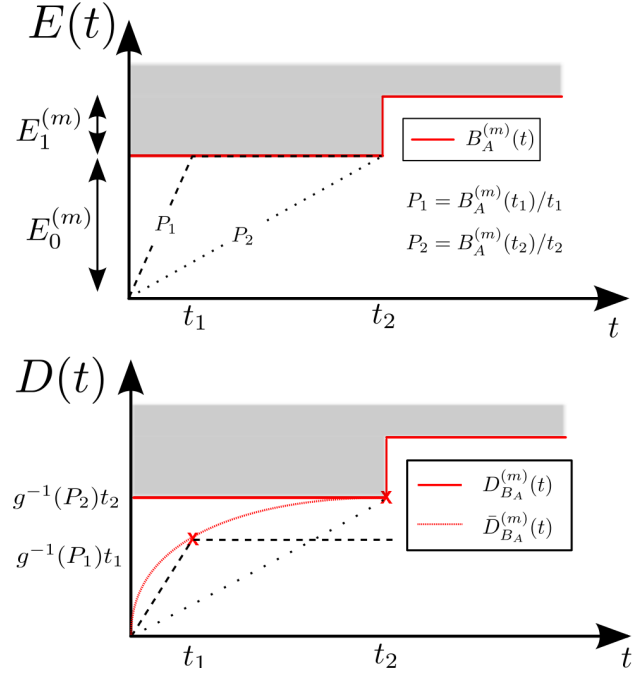}
\vspace{-1em}
\caption{Mapping of the energy constraint to the data domain. The plot in the bottom shows the suboptimality to reach $\bar D^{(m)}_{B_A}(t)$ at any time instant that is not an event. This implies that in practice the \textit{effective mapping} $D^{(m)}_{B_A}(t)$ can be used as the mapping of the energy constraint to the data domain.} 
\label{fig:mappings}
\end{figure}



A similar approach can be done to map the minimum energy expenditure curve $E_{min}^{(m)}(t)$ to the data domain. In this case, since overflows may only be produced at energy arrival events, it is only necessary to map the lower bound in the expended energy to the data domain at these time instants. For the rest of time instants, the time intervals between energy arrivals, a constant value is assigned without loss of generality, hence, obtaining $D_{E_{min}}^{(m)}(t)$.

Figure \ref{fig:EnergyMapping_QOS} shows a representation of the problem once the energy constraint is mapped to the data domain.
Now the problem is simplified, since data and energy  constraints can be merged in a single
constraint that, at every time instant, is the most restrictive of the two
constraints, i.e.,
\begin{equation}
 D_{max}^{(m)}(t) = \min \{D_A^{(m)}(t), D_{B_A}^{(m)}(t) \}.
\end{equation}
Similarly, the lower constraint is 
\begin{equation}
D_{min}^{(m)}(t) = \max \{D_{QoS}^{(m)}(t),
 D_{E_{min}}^{(m)}(t) \}.
\end{equation}

\begin{figure}
\centering
\includegraphics[width=0.75\columnwidth]{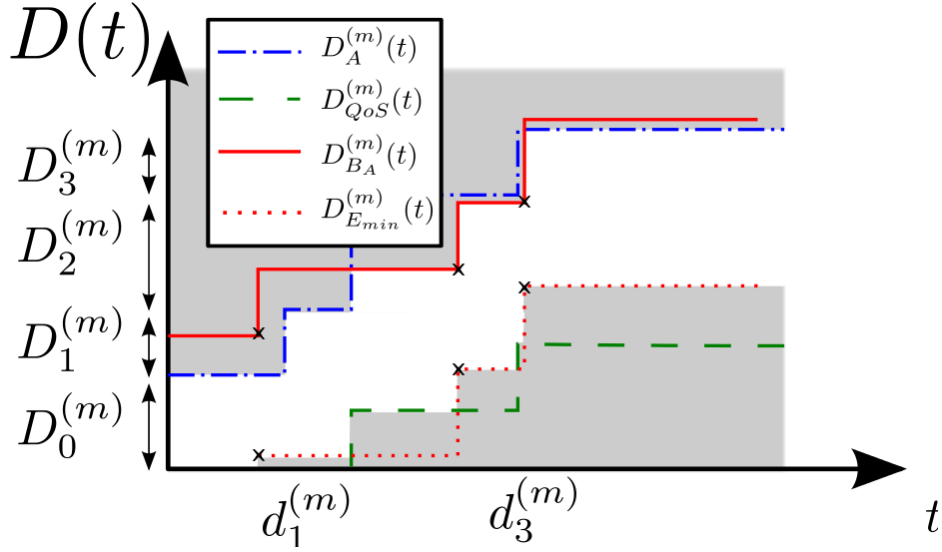}
\caption{Mapping of the \ac{ECC} and minimum energy expenditure to the data domain.}
\label{fig:EnergyMapping_QOS}
\end{figure}




Note that $D_{max}^{(m)}(t)$ and $D_{min}^{(m)}(t)$ are only valid within an algorithm iteration and that it can occur that $D_{max}^{(m)}(t) < D_{min}^{(m)}(t)$. This happens in the following situations: (i.) The node has to transmit a certain amount of data in order not to overflow the battery, however, this data is not still available, i.e, $ D_{E_{min}}^{(m)}(t) > D_A^{(m)}(t)$. (ii.) The node has to transmit a certain amount of data in order to satisfy the \ac{QoS} constraint, however, it does not have enough energy available to do so\ie $D_{B_A}^{(m)}(t) < D_{QoS}^{(m)}(t)$. Note that the situation (ii.) occurs when the problem does not have a feasible solution. As mentioned before, the aim of this section is to model how rate changes are produced when the problem indeed has a solution (at least up to the current algorithm iteration) and, hence, we will focus on situation (i.) where an overflow of the battery is produced.


Let us define the sets of time instants at which $D_{max}^{(m)}(t)$ and $D_{min}^{(m)}(t)$ have discontinuities as $\mathcal{Z}_{max}^{(m)}=\{ t \: | \: D_{max}^{(m)}(t^-) \neq D_{max}^{(m)}(t^+) \}$ and $\mathcal{Z}_{min}^{(m)}=\{ t \: | \: D_{min}^{(m)}(t^-) \neq D_{min}^{(m)}(t^+) \}$, respectively. Remember that, due to Corollary \ref{C:Piecewise_QOS}, we know that $D^\star(t)$ is constant between events defined according to $\mathcal{Z}_{max}^{(m)}$ and $\mathcal{Z}_{min}^{(m)}$. In  Lemmas \ref{L:change_Zmax}, \ref{L:change_Zmin_be}, and \ref{L:change_overflow}, we model the behavior of the optimal solution when the rate changes at a time instant where a single event is produced. Similarly, Lemmas \ref{L:change_Zmax_Zmin_be} and \ref{L:change_Zmax_Zmin_l} describe the behavior of the optimal solution when the rate changes at a time instant where two events are produced. The proofs of these lemmas are given in Appendix \ref{A:rateChanges}.


\begin{lemma}
\label{L:change_Zmax}
If a rate change is produced at a certain time instant $\ell_m$ such that $\ell_m \in \mathcal{Z}_{max}^{(m)}$ and $\ell_m \notin \mathcal{Z}_{min}^{(m)}$, then $D^{\star^{(m)}}(\ell_m) = D_{max}^{(m)}(\ell_m^-)$ and the rate increases, $r_m < r_{m+1}$.
\end{lemma}


\begin{lemma}
\label{L:change_Zmin_be}
If a rate change is produced at the time instant $\ell_m$ such that $\ell_m \notin \mathcal{Z}_{max}^{(m)}$, $\ell_m\in \mathcal{Z}_{min}^{(m)}$ and $D_{max}^{(m)}(\ell_m) \geq D_{min}^{(m)}(\ell_m^+)$, then $D^{\star^{(m)}}(\ell_m) = D_{min}^{(m)}(\ell_m^+)$ and the rate decreases, $r_m > r_{m+1}$.
\end{lemma}

\begin{lemma}
\label{L:change_overflow}
If a rate change is produced at a certain time instant $\ell_m$ such that $\ell_m \notin \mathcal{Z}^{(m)}_{max}$, $\ell_m \in \mathcal{Z}_{min}^{(m)}$ and $D_{max}^{(m)}(\ell_m) < D_{min}^{(m)}(\ell_m^+)$, then an overflow of the battery is produced at  $\ell_m$, $D^{\star^{(m)}}(\ell_m) =D_{max}^{(m)}(\ell_m)$ and the rate is zero until the next data arrival event, $ r_{m+1} = 0$.
\end{lemma}


\begin{lemma}
\label{L:change_Zmax_Zmin_be}
If a rate change is produced at a certain time instant $\ell_m$ such that $\ell_m \in \mathcal{Z}_{max}^{(m)}$ and $\ell_m \in \mathcal{Z}_{min}^{(m)}$, and $D_{max}^{(m)}(\ell_m^-) \ge D_{min}^{(m)}(\ell_m^+)$, then either $D^{\star^{(m)}}(\ell_m) = D_{max}^{(m)}(\ell_m^-)$ and the rate increases, or $D^{\star^{(m)}}(\ell_m) = D_{min}^{(m)}(\ell_m^+)$ and the rate decreases. 
\end{lemma}


\begin{lemma}
\label{L:change_Zmax_Zmin_l}
If a rate change is produced at a certain time instant $\ell_m$ such that $\ell_m \in \mathcal{Z}_{max}^{(m)}$ and $\ell_m \in \mathcal{Z}_{min}^{(m)}$, and $D_{max}^{(m)}(\ell_m^-) < D_{min}^{(m)}(\ell_m^+)$, then an overflow of the battery is produced at  $\ell_m$, $D^{\star^{(m)}}(\ell_m) =D_{max}^{(m)}(\ell_m^-)$ and the rate can either increase or decrease.
\end{lemma}


By using these lemmas, we are able to construct an algorithm, which is presented in next section, that iteratively finds the optimal solution, or concludes that there is no feasible solution.

\section{Optimal data departure curve construction}
\label{sec:algorithm_qos}
In this section, we describe the developed algorithm that is able to either construct $D^\star(t)$ or, alternatively, conclude that the problem in \eqref{Eq:P_EA_DA_Cmax_QOS} does not have a feasible solution. As stated in Corollary \ref{C:Piecewise_QOS}, the optimal data departure curve is a piece-wise linear function. As previously explained, the developed algorithm follows an iterative process where, at the $m$-th iteration, the duration, $\ell_m$, and rate, $r_m$, of an epoch are determined.  We will focus on the explanation of the $m$-th iteration since all the other iterations follow the same approach.

As shown in Figure \ref{fig:block_diagram_alg}, the algorithm is composed by three main blocks. The first block, named $checkSolution$, determines the existence of solution in the current iteration by checking the condition in Lemma \ref{L:Solution}. If the problem does not have a solution, the algorithm ends. Otherwise, the algorithm proceeds to the subsequent blocks to determine the rate and length of the epoch.

\begin{figure}
\centering
\includegraphics[width=0.55\columnwidth]{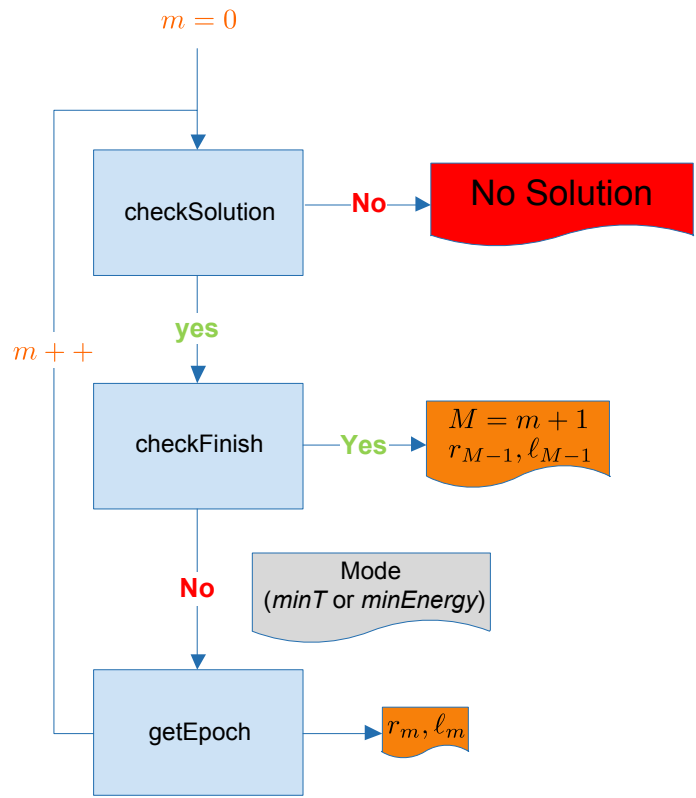}
\caption{Block diagram of the iterative algorithm.}
\label{fig:block_diagram_alg}
\end{figure}

The second block, named \textit{checkFinish}, checks whether it is possible to transmit all the remaining data by using all the available energy at constant rate. This block is necessary because $T^{(m)}$ is unknown. In case it is possible to finish in a single epoch, it is obtained that $T^{(m)} = \ell_m$ and $M = m + 1$. Then, the solution to problem \eqref{Eq:P_EA_DA_Cmax_QOS} has been found and  can be computed as
\begin{equation}
\label{Eq:Solution}
D^\star(t) = \sum_{m=0}^{M-1} r_m (t-\tau_m) \: \Pi \left( \frac{t - \tau_m }{\ell_m}\right) + r_m \ell_m H(t - \tau_{m} -\ell_m),
\end{equation}
where $\Pi(\cdot)$ is the unit pulse in the interval $[0,1]$. The domain of $D^\star(t)$ is $[0, T = \sum_{m=0}^{M-1} \ell_m]$.

Otherwise, if it is not possible to finish in a single epoch, the block  \textit{checkFinish} returns the mode ($minT$ or $minEnergy$) to be used by the third block, which we name \textit{getEpoch}, to find the epoch rate and length ($r_m$ and $\ell_m$) that fulfill Lemmas \ref{L:change_Zmax}-\ref{L:change_Zmax_Zmin_l}. The mode $minT$ is used when the node already has enough energy to finish transmission, whereas, the mode $minEnergy$ is used when the node is still not able to finish transmission at any rate and, hence, the objective is to save as much energy as possible for the end of the transmission. An extended explanation of the inner behavior of each of these blocks is given in Appendix \ref{Ap:Alg}. 

Once $r_m$ and $\ell_m$ are determined, the origin of coordinates is moved to the point $(\ell_m, r_m \ell_m)$ and the variables are prepared for the new iteration. In the data domain, the iteration transmitted bits $D^{\star^{(m)}}(\ell_m) =r_m \ell_m$ are subtracted from $D_A^{(m)}(t)$ and $D_{QoS}^{(m)}(t)$, for instance,  $D_A^{(m+1)}(t) = D_A^{(m)}(t + \ell_m) - D^{\star^{(m)}}(\ell_m)$. Similarly, in the energy domain, the expended energy $E^{\star^{(m)}}(\ell_m) = g(r_m) \ell_m$ is subtracted from $B_A^{(m)}(t)$ and $E_{min}^{(m)}(t)$. Moreover, in case that transmitting at $r_m$ produces a battery overflow at time instant $\ell_m$, the amount of energy lost due to the overflow is also subtracted from these variables. Finally, the mapping to the data domain is recalculated for the iteration $m+1$ and the whole procedure starts again to determine $r_{m+1}$ and $l_{m+1}$.

There are two possible reasons for which the algorithm ends: (i.) At some iteration, Lemma \ref{L:Solution} is satisfied and, hence, the problem does not have solution. (ii) All the data has been transmitted and the optimal data departure curve has been obtained as given in \eqref{Eq:Solution}.


The algorithm optimality is summarized in the following theorem and its subsequent proof:
\begin{theorem}
\label{TH:Algorithm_QOS}
The algorithm presented in this section constructs the optimal data departure curve, $D^\star(t)$, for the problem \eqref{Eq:P_EA_DA_Cmax_QOS}.
\end{theorem}

\begin{proof}
See appendix \ref{A:Proof_Th_QOS}.
\end{proof}

\vspace{-0.5cm}
\section{Results}
To the best of our knowledge there is no other algorithm in the literature that considers altogether the \ac{ECC}, \ac{DCC}, the \ac{QoS} constraint and the finite battery capacity. Therefore, in order to get some insights on the gain obtained with our proposed solution, we have developed a suboptimal ad-hoc strategy, namely, the \ac{EBS}, that tries to empty the buffers as soon as possible\ie it looks for the time instant at which the next arrival (energy or data) is produced and tries to transmit at a constant rate so that the corresponding buffer is emptied by the time of the corresponding arrival.

The left $y$-axis of Figure \ref{fig:Comparison} compares the normalized mean minimum $T$ along a total of $1000$ iterations, where at each iteration the data and energy arrivals are randomly generated following a uniform distribution.
The amount of energy in each of the packets is normalized according to the total harvested energy which varies along the $x$-axis.
The right $y$-axis shows the percentage of iterations in which there exists a feasible solution to \eqref{Eq:P_EA_DA_Cmax_QOS}.
As shown in Figure \ref{fig:Comparison}, our proposed optimal algorithm substantially reduces the mean minimum $T$. 
If feasible solutions exist the optimal algorithm finds the one that minimizes $T$, however, for some arrival profiles, the \ac{EBS} is not able to find any feasible solution in spite of its existence.

\begin{figure}
\centering
\includegraphics[width=0.75\columnwidth]{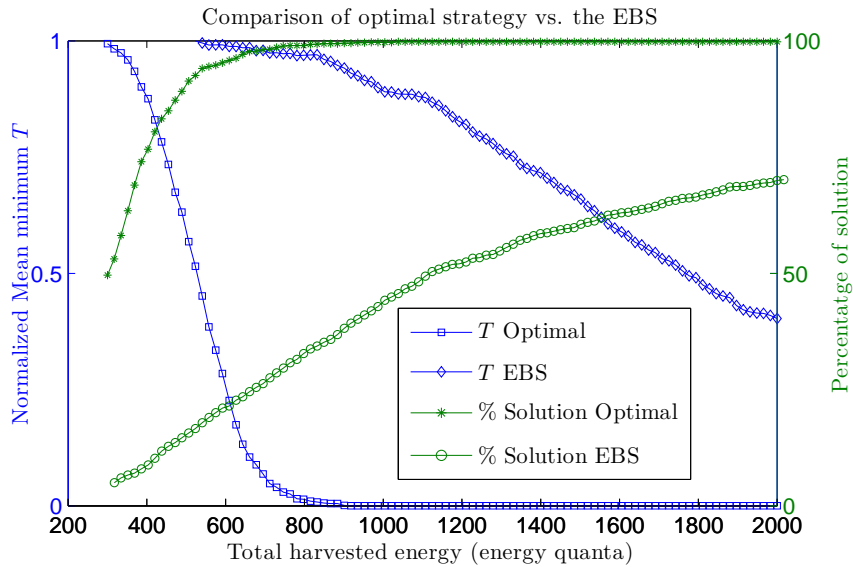}
\caption{Performance comparison of our proposed optimal algorithm with respect to the suboptimal \ac{EBS} in terms of mean minimum completion time and percentage of solutions. The lines marked with rectangles and diamonds refer to the left axis, whereas the lines marked with circles and asterisks refer to the rights axis.}
\label{fig:Comparison}
\end{figure}

Note that the \ac{EBS} performs worse than the optimal strategy because it changes the rate at every packet arrival without checking whether constant power transmission is feasible between two arrivals, which would consume less energy.
As a result of this extra energy consumption, the mean minimum $T$ is higher and the probability of finding a feasible solution decreases.

\section{Conclusion}
\label{sec:conclusions}
The presence of energy harvesters in wireless nodes implies a loss of optimality of the traditional transmission policies such as the well-known water-filling strategy \cite{cover_elementsIT_1991}. In this paper, the optimal transmission strategy has been obtained for wireless energy harvesting nodes with finite battery capacity that, additionally, have to fulfill some \ac{QoS} constraint. Hence, we have contributed into the decrease of the required transmission completion time and, thus, increased the overall efficiency in the use of the harvested energy. Moreover, in more technical terms, we have seen that, as far as battery does not overflow, constant rate transmission is the strategy that requires less energy to transmit a certain amount of data. However, if indeed battery overflows, transmitting at constant rate is not optimal anymore, but the optimal strategy increases the rate before the overflow until either there is no overflow, or the \ac{DCC} is reached (i.e., there is no more data to transmit). We have seen that the existence of the optimal solution depends both on the dynamics of the harvesting process and on the required \ac{QoS}. According to this, we have developed an algorithm that is able to determine whether the problem has a solution or not and, in case of having a solution, determines the optimal data transmission strategy.

\vspace{-0.3cm}
\appendices
\vspace{-0.1cm}
\section{}
\label{A:rateChanges}
\vspace{-0.1cm}
\subsection*{Proof of Lemma \ref{L:change_Zmax}}
\label{A:change_Zmax}
The proof of the lemma is divided in two parts. We first show that if a rate change is produced at $\ell_m$, then $D^{\star^{(m)}}(\ell_m) = D_{max}^{(m)}(\ell_m^-)$ and, afterwards, we show that the rate must increase.
\begin{proofPart}
\label{PF:4_1}
Let us assume that a rate change is produced at $\ell_m$ such that $D^{(m)}(\ell_m) <  D_{max}^{(m)}(\ell_m^-)$. We will show by contradiction that this cannot be optimal. Let us consider the time interval $(\ell_m -\epsilon, \ell_m + \epsilon)$ with $\epsilon$ being positive. Note that if we select a sufficiently 
small $\epsilon$, we can find a straight line with rate $ r = \frac{D^{(m)}(\ell_m +\epsilon)-D^{(m)}(\ell_m -\epsilon)}{2\epsilon}$, which still satisfies the constraints, that transmits the same amount of data while having less energy expenditure. Hence, we have proved that if the rate changes at $\ell_m$, then $D^{\star^{(m)}}(\ell_m) = D_{max}^{(m)}(\ell_m^-)$.
\end{proofPart}
\begin{proofPart}
\label{PF:4_2}
Now we prove that when the rate changes at $\ell_m$, it must increase. The procedure is the same as the one in the first part of the proof. We start by assuming that a rate decrease is optimal and then, we see that it leads to a logical contradiction. We denote $r_m$ and $r_{m+1}$ the rates before and after $\ell_m$, respectively, where $r_m > r_{m+1}$. We consider the same time interval. In this case, we can also find a straight line whose slope is $r = \frac{D^{(m)}(\ell_m +\epsilon)-D^{(m)}(\ell_m -\epsilon)}{2\epsilon}$, which satisfies energy and data constraints, that transmits the same amount of data while having less energy expenditure. Hence, we have proved by contradiction that if the rate changes, it must increase.
\end{proofPart}

\vspace{-0.3cm}
\subsection*{Proof of Lemma \ref{L:change_Zmin_be}}
\label{A:change_Zmin_be}
This proof is similar to the proof of Lemma \ref{L:change_Zmax}. The main difference is that now the discontinuity is in the lower constraint. Then, by following the same procedure we can first show that a rate change is suboptimal unless $D^{\star^{(m)}}(\ell_m) = D_{min}^{(m)}(\ell_m^+)$ and a rate decrease is produced.

\vspace{-0.3cm}
\subsection*{Proof of Lemma \ref{L:change_overflow}}
\label{A:Zmin_l}
We know that $\ell_m$ comes from an energy arrival event, otherwise, the problem would not have a feasible solution as stated in Lemma \ref{L:Solution}.
This implies that $D_{max}^{(m)}(\ell_m)=D_{A}^{(m)}(\ell_m)$ and $D_{min}^{(m)}(\ell_m)=D_{E_{min}}^{(m)}(\ell_m^+)$. Then, it is obvious that if $D_{A}^{(m)}(\ell_m) < D_{E_{min}}^{(m)}(\ell_m^+)$, an overflow is produced. Hence, the proof that $D^{\star^{(m)}}(\ell_m) =D_{max}^{(m)}(\ell_m)$ comes from Lemma \ref{Result:NoOverflow}, where we show that the optimal data departure curve satisfies that when an overflow is produced, all the data has been transmitted and, hence, the rate must change to zero until the following data arrival.

\vspace{-0.3cm}
\subsection*{Proof of Lemma \ref{L:change_Zmax_Zmin_be}}
\label{A:change_Zmax_Zmin_be}
This lemma states that two events are produced at time $\ell_m$. One could look at this lemma as the union of Lemmas \ref{L:change_Zmax} and \ref{L:change_Zmin_be} at the same time instant, hence, the proof has been already given in the aforementioned lemmas.

\vspace{-0.3cm}
\subsection*{Proof of Lemma \ref{L:change_Zmax_Zmin_l}}
\label{A:change_Zmax_Zmin_l}
As we pointed out, this lemma only applies when the problem has a solution. Consequently, we know that $D_{min}^{(m)}(\ell_m) = D_{E_{min}}^{(m)}(\ell_m^+)$ and that $D_{max}^{(m)}(\ell_m)=D_A^{(m)}(\ell_m^-)$. Hence, this is the overflow problem with the particularity that a data arrival is produced at $\ell_m$. From Lemma \ref{Result:NoOverflow}, the optimal solution minimizes the energy lost due to overflow and, thus, $D^{\star^{(m)}}(\ell_m) =D_{max}^{(m)}(\ell_m^-)$. However, in this situation, nothing can be stated regarding the rate in the following epoch.
\vspace{-0.1cm}
\section{}
\label{Ap:Alg}
In this appendix, a technical explanation of the algorithm is given. First, we introduce the maximum and minimum rates, a concept required to understand the second and third blocks of the algorithm that are presented afterwards.

\vspace{-0.3cm}
\subsection{Maximum and minimum rates}
\label{Sec:RateComputation}
Let $\mathbb{R}_{max}^{(m)}$ denote the set that contains the rates obtained by joining the reference point, i.e., $(0,D^{(m)}(0)=0)$\footnote{Note that the reference point is always $(0,0)$ as the origin of coordinates is moved at every iteration.}, with the discontinuities from the left of $D_{max}^{(m)}(t)$\ie the points $(z, D_{max}^{(m)}(z^-))$, $\forall z \in \mathcal{Z}_{max}^{(m)}$,  and such that the obtained curve is feasible for $t \in (0,z)$, where by feasible we mean that the curve satisfies all the constraints. Similarly, $\mathbb{R}_{min}^{(m)}$ contains the rates obtained from joining the reference point with the discontinuities from the right of $D_{min}^{(m)}(t)$\ie the points $(z, D_{min}^{(m)}(z^+))$, $\forall z \in \mathcal{Z}_{min}^{(m)}$, and such that the obtained curve is feasible in the same interval. An example of this can be seen in Figure \ref{fig:Algorithm}.

\begin{figure}
\centering
\includegraphics[width=0.75\columnwidth]{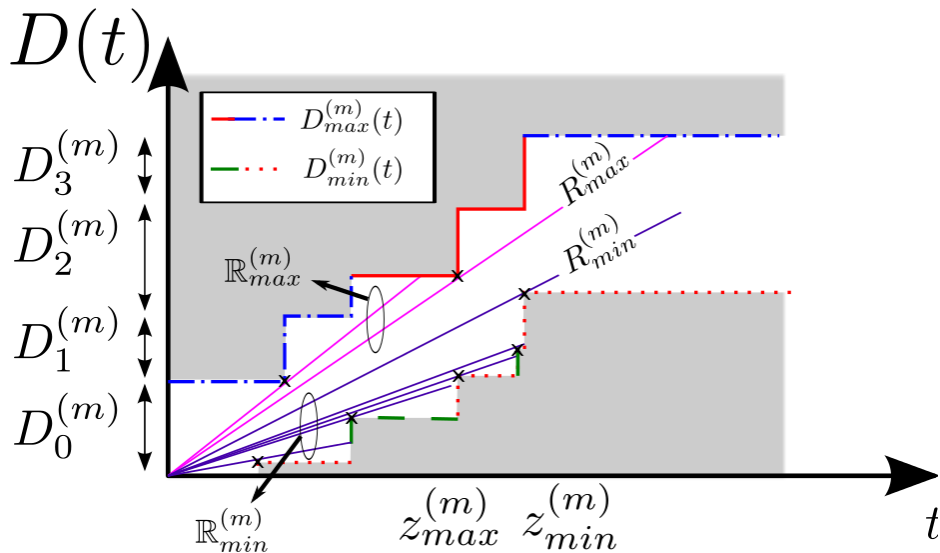}
	\caption{Graphical interpretation of the computation $R_{min}^{(m)}$ and $R_{max}^{(m)}$. The
	optimal epoch rate lies within the interval $[R_{min}^{(m)}, R_{max}^{(m)}]$.}
	\label{fig:Algorithm}
	\end{figure}
	
	Let $R_{max}^{(m)}$ denote the infimum of the set $\mathbb{R}_{max}^{(m)}$ and $R_{min}^{(m)}$  refer to the supremum of the set  $\mathbb{R}_{min}^{(m)}$. Let $z_{max}^{(m)}$ and $z_{min}^{(m)}$ denote the time instants from which $R_{max}^{(m)}$ and $R_{min}^{(m)}$ have
	been obtained.
Then, all the rates above $R_{max}^{(m)}$ and below $R_{min}^{(m)}$ are suboptimal as they
	would require a rate change to transmit the same amount of data.\footnote{Note that $R_{max}^{(m)}$ is always greater than  $R_{min}^{(m)}$, otherwise,
	either $R_{max}^{(m)}$ or $R_{min}^{(m)}$ would not be feasible.}
 Thus, the
	optimal rate lies within the interval $[R_{min}^{(m)}, R_{max}^{(m)}]$. 

	\subsection{Finish transmission at a constant rate (checkFinish)}
	\label{SS:FinishEven}
	The first step, which is presented in Function \ref{ALG:Finish}, checks whether it is possible to transmit all bits by using an even power allocation in just one epoch. If it is possible, which implies that transmission is ended, the algorithm returns the rate and length of the last epoch, otherwise, it returns the strategy or mode that will be used in order to determine the following epoch. 
	
	The function \textit{checkFinish} first checks whether by transmitting at the maximum feasible rate,
	$R_{max}^{(m)}$, it is possible to transmit all the remaining data $D_{Tot}^{(m)}$ (this is done by the subroutine $getDataInCrossing$). In case
	it is not possible, the function returns the mode $minEnergy$. Otherwise, the
	function finds the time $\hat{T}^{(m)}_0$ required to transmit the remaining data
	$D_{Tot}^{(m)}$ with the iteration initial battery $E_0^{(m)}$.\footnote{Remember that $E_0^{(m)} = B_A(\tau_m;\tau_m) - E^\star(\tau_m)$ as summarized in Table \ref{tab:Notation}.}
 Then, it computes the equivalent rate $\hat R^{(m)}$ and
	checks whether transmitting at this rate is feasible, i.e., the following two
	conditions are fulfilled: (i.) $\hat{R}^{(m)} \leq R_{max}^{(m)}$ and (ii.) $\hat{R}^{(m)} \geq
	R_{min}^{(m)}$. In case that (i.) is not fulfilled, the function returns the mode
	$minT$. If (ii.) is not met, it is checked if, by using the following energy
	arrivals, a feasible curve is obtained. Finally, in case both conditions are
	fulfilled, it is checked  whether any energy arrival has been produced in the
	time interval $(0,\hat{T}^{(m)}_0)$. In case 
	of
	no arrivals, the algorithm ends
	and the last epoch has been found. In case there is an energy arrival in
	$(0,\hat{T}^{(m)}_0)$, the function repeats the whole process but now using the
	initial battery, $E_0^{(m)}$, and the energy of the first arrival, $E_1^{(m)}$. This process is
	repeated until (i.) becomes false or a feasible curve is found. 
		
\vspace{-0.4cm}
	\subsection{Get rate and length of the next epoch (getEpoch)}
	\label{SS:GetEpoch}
	This algorithm's block uses the parameter $mode$, which is obtained from the function $checkFinish$ as presented in Appendix \ref{SS:FinishEven}, to compute the rate and length of the following epoch. 
	
	\textit{Minimize the total completion time $(mode == minT)$:} This strategy is used when both of the following conditions are satisfied: (i.) It is possible to finish the transmission at some rate $r$ with $r \leq R_{max}^{(m)}$. (ii.) The rate obtained from an even power allocation $\hat{R}^{(m)}$ is not feasible due to $\hat{R}^{(m)} > R_{max}^{(m)}$. Hence, the objective is to find the rate that allows us to finish transmission as soon as possible, without paying attention on saving power, however, without wasting it, either. In such case, the rate and duration of the epoch are $r_m=R_{max}^{(m)}$ and $\ell_m = z_{max}^{(m)}$.
	
	\textit{Minimize the energy expenditure $(mode == minEnergy)$:} This strategy is used when, due to the constraints,  transmission cannot be  finished at constant rate, e.g., the rate must be increased to satisfy the \ac{QoS} constraints. Hence, the objective is to save as much energy as possible in order to use it when ending the transmission is feasible. Note that in this situation, the problem of obtaining the following epoch is similar to the problem presented in \cite{zafer_calculus_2009} and, hence, the solution is also similar. The possible data departure curves with constant rate $r$, i.e., $D^{(m)}(t)=rt$, are divided in two sets. The first set $\mathbb{S}_{R_{max}}^{(m)}$ contains all the rates $r$ such that the associated data departure curve crosses the constraint $D_{max}^{(m)}(t)$ first. Whereas the set $\mathbb{S}_{R_{min}}^{(m)}$ contains all the rates $r$ such that the associated data departure curve crosses the constraint $D_{min}^{(m)}(t)$ first. Then, the rate of the following epoch is determined as the infimum of $\mathbb{S}_{R_{max}}^{(m)}$ or, equivalently, the supremum of $\mathbb{S}_{R_{min}}^{(m)}$, i.e., $r_m=\inf \left(\mathbb{S}_{R_{max}}^{(m)} \right) = \sup \left(\mathbb{S}_{R_{min}}^{(m)}\right)$
	and the duration of the epoch $\ell_m$ can be obtained as the first time instant such that,
	$r_m \ell_m = D_{max}^{(m)}(\ell_m)$ or $r_m \ell_m = D_{min}^{(m)}(\ell_m)$.

	\setcounter{proofPart}{0}
	\vspace{-0.3cm}
	\subsection{Proof of the algorithm optimality}
	\label{A:Proof_Th_QOS}
	At each iteration, the algorithm checks whether Lemma \ref{L:Solution} is fulfilled. In such a case, there is no solution for the problem and the algorithm ends. Otherwise, the algorithm must satisfy Lemmas \ref{L:change_Zmax}-\ref{L:change_Zmax_Zmin_l} at each rate change and Lemma \ref{L:useAllEnergy} by the end of the transmission. 
	To show that these lemmas are satisfied and, hence, the algorithm computes the optimal data departure curve,
	we focus on the three different situations that can occur depending on the
	constraints of the problem: (i.) The algorithm finishes transmission by using an
	even power allocation among all bits to be transmitted. (ii.) $minT$ strategy is used
	to obtain the epoch. (iii.) The epoch is computed by using $minEnergy$ mode.
	In the following, the optimality of 
	these cases
	is proved by showing
	that the
	algorithm-chosen solution
	satisfies the optimality conditions
	and that it is unique.
	
	\begin{proofPart}[Even power allocation]
	When this situation occurs, the algorithm ends transmission by transmitting at constant rate. Hence, Lemma \ref{L:NoEvents_QOS} is satisfied. No overflow is produced, thus, Lemma \ref{Result:NoOverflow} is satisfied. Lemmas from \ref{L:change_Zmax} to \ref{L:change_Zmax_Zmin_l} do not apply since there are no rate changes. Finally, note that Lemma \ref{L:useAllEnergy} is satisfied, since the last rate $r_{M-1}$ is obtained as the rate that allows the transmission of all the data by using all the available energy, moreover, from the properties of the function $g(\cdot)$ this rate is unique.
	
	
	\end{proofPart}
	
	\begin{proofPart}[minT mode]
	%
	
	We want to demonstrate that the optimal departure curve transmits at $R_{max}^{(m)}$ during a period of time of $z_{max}^{(m)}$. Let $\hat{T}^{(m)}$ be the total completion time that would be obtained if transmitting at rate $\hat{R}^{(m)}$ was feasible, hence, $\hat{T}^{(m)} = D_{Tot}^{(m)}/\hat{R}^{(m)}$. Similarly, $T_{max}^{(m)}= D_{Tot}^{(m)}/R_{max}^{(m)}$. Note that, since $R_{max}^{(m)}$ is a feasible rate, $T_{max}^{(m)}$ is an upper bound of the remaining completion time, $T^{(m)}$, whereas $\hat{T}^{(m)}$ is a lower bound, hence:	
	\begin{equation}
	\hat{T}^{(m)} <  T^{(m)} < T_{max}^{(m)}.
	\label{Eq:T_boundary}
	\end{equation}
	
	Consider the data departure curve $D_1^{(m)}(t) = R_{max}^{(m)}t$. Note that any other data departure curve, $D_2^{(m)}(t)$, is suboptimal since, in order to satisfy \eqref{Eq:T_boundary}, $D_2^{(m)}(t)$ will cross $D_1^{(m)}(t)$ for some $t_y \in (0,T_{max}^{(m)})$. Hence, at $t_y$, both curves have sent the same amount of data, however, $D_1^{(m)}(t)$ has consumed less energy.
	%
	
	
	Now we must show that, at $z_{max}^{(m)}$, the rate increases. Note that by transmitting at $R_{max}^{(m)}$ instead of at $\hat{R}^{(m)}$, some energy has been saved. Then, in the following epoch, the available energy per bit is higher and, then, the new rate $\hat{R}^{(m+1)}$ obtained from an even power allocation among all bits in the following epoch fulfills $\hat{R}^{(m+1)} > \hat{R}^{(m)} > R_{max}^{(m)}$. Hence, we have proved that at $z_{max}^{(m)}$ a rate increase is produced.

	\end{proofPart}
	
	\begin{proofPart}[minEnergy mode]
	This mode is used when it is not possible to finish transmission at any rate. Note that the algorithm can select three different kind of points, denoted as $v_x = (\ell_m, D^{\star^{(m)}}(\ell_m))$, for ending the epoch depending on the constraints:
	
	\begin{itemize}
	\item $v_1 \: | \: D^{\star^{(m)}}(\ell_m)=D_{max}^{(m)}(\ell_m^-)$ where $\ell_m$ is either $d_i^{(m)}$ or $e_j^{(m)}$.
	\item $v_2 \: | \: D^{\star^{(m)}}(\ell_m)=D_{min}^{(m)}(\ell_m^+)$ where $\ell_m$ is either $q_k^{(m)}$ or $e_j^{(m)}$.
	\item $v_3 \: | \:D^{\star^{(m)}}(\ell_m)=D_{max}^{(m)}(\ell_m^-)$ where $\ell_m$ is $e_j^{(m)}$ and an overflow is produced.
\end{itemize} 

Note that Corollary \ref{C:Piecewise_QOS} and Lemma \ref{Result:NoOverflow} are satisfied for any of the selected points and that Lemma \ref{L:useAllEnergy} does not apply since transmission cannot be ended, yet. Hence, we have to prove the following three conditions: (i.) If a point such as $v_1$ is selected, there will be a rate increase (Lemma \ref{L:change_Zmax} or Lemma \ref{L:change_Zmax_Zmin_be} for the case that at $\ell_m$ two events are produced). (ii.) If a point such as $v_2$ is selected, there will be a rate decrease (Lemma \ref{L:change_Zmin_be} or Lemma \ref{L:change_Zmax_Zmin_be} for the case that at $\ell_m$ two events are produced). And (iii.), if a point such as $v_3$ is selected, the rate of the following epoch is zero as far as $\ell_3$ is not a data arrival event (Lemma \ref{L:change_overflow} and Lemma \ref{L:change_Zmax_Zmin_l}). 

Regarding (i.), the rate of the epoch is $r_m = \frac{D_{max}^{(m)}(\ell_m^-)}{\ell_m}$ that is the supremum of $\mathbb{S}_{R_{min}}^{(m)}$. Note that, at the following iteration, the set $\mathbb{S}^{(m + 1)}_{R_{min}}$ includes all the rates in the interval $(0 , r_m)$ and the rates contained in the interval $(r_m, r_m + \epsilon)$, for some $\epsilon > 0$. Hence, the rate of the following iteration, $r_{m+1}$, satisfies that $r_{m+1} \geq r_m$, therefore, a rate increase is produced. A similar approach can be done for (ii.), the rate of the epoch is  $r_m =\frac{D_{min}^{(m)}(\ell_m^+)}{\ell_m}$ that is the infimum of $\mathbb{S}_{R_{max}}^{(m)}$. At the following iteration the set $\mathbb{S}^{(m +1)}_{R_{max}}$ contains all the rates in $(r_m - \epsilon , \infty)$ for some $\epsilon > 0$. Hence, the rate of the following iteration, $r_{m+1}$, satisfies that $r_{m+1} \leq r_m$ and, therefore, there is a  rate decrease. Finally, in case (iii.), an overflow of the battery is produced. Note that the solution chosen by the algorithm satisfies that all the available data at $\ell_3^-$ has been transmitted. This implies that if non-data arrival is produce at $\ell_3$ the rate of the following epoch must be zero in order to satisfy data causality constraint. 

Therefore, the algorithm computes the optimal solution since it satisfies all the lemmas that model the behavior of the optimal solution.


\end{proofPart}

\bibliographystyle{IEEEtran}
\small{
\bibliography{Journal_QoS}
}
\newpage

\pagebreak[4]
\begin{table}[H]
\centering
\begin{tabular}{|l|c|c|}
\hline
Definition & General notation & Notation at the $m$-th algorithm iteration \\
\hline
Instantaneous power 				& $P(t)$ 			& \\
Instantaneous rate 				& $r(t)$ 			&\\
Power-rate function 				&$g(\cdot)$ 		& \\
Transmission completion time  & $T$					& $T^{(m)} =  T -\tau_m$			\\
\hline
Data departure curve 			& \DDC 			& $ D^{(m)}(t)= D(t+ \tau_m) -  D^\star(\tau_m)$ \\
Optimal data departure curve 			& $D^\star(t)$ 			& $ D^{\star^{(m)}}(t) = D^\star(t+ \tau_m) -  D^\star(\tau_m) $ \\
Accumulated Data 					&\acData 			& $\acDataIt = \acData[t + \tau_m] -  D^\star(\tau_m) $\\
Minimum Data Departure 			&\minDDep 			& $\minDDepIt = \{\minDDep[t+\tau_m] -  D^\star(\tau_m) \}^+$ \\
\hline
Energy expenditure curve 		& $\EEC$			& $\EEC[(m)]= E(t+ \tau_m) -  E^\star(\tau_m)$\\
Optimal energy expenditure curve 		& $E^\star(t)$			& $ E^{\star^{(m)}}(t) = E^\star(t+ \tau_m) -  E^\star(\tau_m) $\\
Accumulated Battery 				&\acBatery[t]{t_x}		 	& $\acBateryIt = \acBatery[t + \tau_m]{\tau_m} - E^\star(\tau_m) $ \\
Minimum Energy Expenditure 	& \minEExp{t_x} 			& $\minEExpIt = \{ \minEExp[t+\tau_m]{\tau_m} - E^\star(\tau_m) \}^+$ \\
\hline
Actual mapping of $\acBatery[t]{\tau_m}$ to data domain 	& & $\bar D_{B_A}^{(m)}(t)$ \\
Effective mapping of $\acBatery[t]{\tau_m}$ to data domain 	& & $D_{B_A}^{(m)}(t)$ \\
Effective mapping of $\minEExp[t]{\tau_m}$ to data domain 	& & $D_{E_{min}}^{(m)}(t)$ \\
\hline
Equivalent upper bound on the data domain & & $D_{max}^{(m)}(t) = \min \{D_A^{(m)}(t), D_{B_A}^{(m)}(t) \}$ \\
Equivalent lower bound on the data domain & & $D_{min}^{(m)}(t) = \max \{D_{QoS}^{(m)}(t), D_{E_{min}}^{(m)}(t) \}$ \\
Discontinuities of the upper bound & & $\mathcal{Z}_{max}^{(m)}=\{ t \: | \: D_{max}^{(m)}(t^-) \neq D_{max}^{(m)}(t^+) \}$ \\ 
Discontinuities of the lower bound & & $\mathcal{Z}_{min}^{(m)}=\{ t \: | \: D_{min}^{(m)}(t^-) \neq D_{min}^{(m)}(t^+) \}$ \\
\hline
$m$-th epoch rate					&						& $r_m$			\\
$m$-th epoch length				&						& $\ell_m$			\\
Beginning $m$-th epoch			&						& $\tau_m$			\\
\hline
Data arrival time					& \DArrivalT		& $\DArrivalT[(m)]= \DArrivalT -\tau_m$			\\
\multirow{2}{*}{Amount of data in the packet}					& \multirow{2}{*}{\DArrival} 		& $D_0^{(m)} = \acData[\tau_m] - D^\star(\tau_m)$\\  
& & $D_i^{(m)}$ is a relabeling of $D_i$ for $d_i>\tau_m$  			\\
\hline
Energy arrival time				& \EArrivalT		& $\EArrivalT[(m)] = e_j - \tau_m$			\\
\multirow{2}{*}{Amount of energy in the packet}				& \multirow{2}{*}{\EArrival}		& $E_0^{(m)} = \acBatery[\tau_m]{\tau_m} - E^\star(\tau_m$)			\\
& & $E_j^{(m)}$ is a relabeling of $E_j$ for $e_j>\tau_m$  			\\
\hline
\ac{QoS} requirement arrival time	& \QArrivalT		& $\QArrivalT[(m)] = q_k -\tau_m$			\\
\multirow{2}{*}{Amount of data in the \ac{QoS} requirement} & \multirow{2}{*}{\QArrival}	&  $Q_0^{(m)} = 0$\\
& & \QArrival[(m)]  is a relabeling of $Q_k$ for $q_k>\tau_m$			\\
\hline
\end{tabular}
\caption{Summary of the paper notation.}
\label{tab:Notation}
\end{table}

\pagebreak[4]
\newcounter{temp}
	\setcounter{temp}{\value{algorithm}}
	\setcounter{algorithm}{0}
	\floatname{algorithm}{Function}

	\begin{algorithm}[H]
	\caption{checkFinish}\label{ALG:Finish}
	\begin{algorithmic}[0]
	
	
	\State $D_{Tot}^{(m)} = D_{max}^{(m)}(\infty)$
	\Comment Remaining data 
	\If {$(getDataInCrossing(\: R_{max}^{(m)} t \: ,\: D_{max}^{(m)}(t)\:) < D_{Tot}^{(m)})$}
		\State \Return $mode = minEnergy$, $finish = 0$
		\Comment It is not possible to finish yet.
	\Else
		\For {$i=0:J^{(m)}-1$}
	      \Comment $J^{(m)}$ is the number of packets with $e_j > \tau_m$ plus one.
			\State $E = \sum_{j=0}^{i} E_j^{(m)}$
			\State  $\hat{T}^{(m)}_i$ is obtained by solving  $g^{-1}(E/\hat{T}^{(m)}_i) = D_{Tot}^{(m)}/\hat{T}^{(m)}_i$ 
			
			\State $\hat{R}^{(m)}=D_{Tot}^{(m)}/\hat{T}^{(m)}_i$
			\Comment Even power allocation among all bits
			\If {$R_{min}^{(m)}\leq \hat{R}^{(m)} \leq R_{max}^{(m)}$}
				\Comment $D(t)= \hat{R}^{(m)}t$ is feasible in $[0,\hat{T}^{(m)}_i]$
				\State $\mathbb{S} = \{ E^{(m)}_j \: | \: j > i \:, \: e_j^{(m)} \in (0,\hat{T}^{(m)}_i) \} $
				\If {$\mathbb{S} = \emptyset$}
					\Comment The algorithm ends and the rate and length of the last epoch are returned.
					\State \Return $r= \hat{R}^{(m)}$, $\ell = \hat{T}^{(m)}_i$, $finish = 1$
				\EndIf
			\ElsIf {$(\hat{R}^{(m)} > R_{max}^{(m)})$}
				\State \Return $mode = minT$, $finish = 0$ 
			\EndIf 
		\EndFor
		\State \Return $mode = minEnergy$, $finish = 0$
	\EndIf
	
	\end{algorithmic}
	\end{algorithm}

\end{document}